\documentclass[letterpaper,USenglish,cleveref,numberwithinsect,thm-restate]{no-lipics-v2022}
\usepackage[utf8]{inputenc}
\usepackage{color}
\usepackage{csquotes}
\usepackage{xspace}
\usepackage{amsmath}
\usepackage{amssymb}
\usepackage{mathtools}
\usepackage{bbm}
\usepackage{bm}
\usepackage[linesnumbered, noend]{algorithm2e}
\usepackage[draft]{fixme}
\usepackage{tikz}
\usetikzlibrary{arrows,arrows.meta,decorations.pathreplacing,decorations.pathmorphing,shapes,calc,patterns,shapes,matrix,math,quotes,positioning}
\usepackage{upgreek}

\nolinenumbers

\newcommand{\Oh}{\ensuremath{\mathcal{O}}\xspace}

\newcommand{\set}[1]{\{#1\}}
\newcommand{\setof}[2]{\set{#1\colon\,#2}}

\newcommand{\mycomment}[1]{}

\newcommand{\Min}{\textsc{Min}\xspace}
\newcommand{\Max}{\textsc{Max}\xspace}
\newcommand{\bp}{\operatorname{bp}}
\newcommand{\val}{\operatorname{val}}

\title{Bottleneck Paths Reduce to Deterministic Graphical Games and a Counterexample to a Claimed Linear-Time Algorithm}
\titlerunning{Bottleneck Paths Reduce to Deterministic Graphical Games}

\author{Egor Gorbachev}{ETH Zürich and Max Planck Institute for Informatics}{peltorator@pm.me}{https://orcid.org/0009-0005-5977-7986}{}

\authorrunning{E. Gorbachev}

\Copyright{Egor Gorbachev}

\acknowledgements{The author thanks László Kozma and Uri Zwick for helpful discussions.}


\begin{document}

\maketitle

\begin{abstract}
Chechik, Kaplan, Thorup, Zamir, and Zwick (STACS 2016) claimed a simple deterministic linear-time comparison-based algorithm for solving deterministic two-player, turn-based, zero-sum terminal-payoff games, also known as deterministic graphical games (DGGs).
We give a counterexample to their algorithm.

We also give a deterministic linear-time reduction from the directed $s$--$t$ bottleneck path (BP) problem to the DGG problem.
Consequently, a linear-time comparison-based algorithm for computing the value of a designated start vertex in a DGG would yield a linear-time comparison-based algorithm for directed $s$--$t$ BP.
Whether directed $s$--$t$ BP admits such an algorithm has remained open since Gabow and Tarjan gave their $\Oh(m\log^* n)$-time algorithm.
Thus, a positive resolution of the open question for DGGs would also resolve the corresponding open question for directed $s$--$t$ BP.

\end{abstract}

\section{Introduction}
\label{sec:introduction}

In the directed $s$--$t$ bottleneck path (BP) problem, the input is a finite directed graph $G=(V,E)$ with $n\coloneqq |V|$ vertices and $m\coloneqq |E|$ edges, a weight function $w\colon E\to\Lambda$ into a totally ordered set $\Lambda$, and two distinguished vertices $s$ and $t$.
The goal is to find an $s$--$t$ path whose largest edge weight is as small as possible.
The problem appears in the work of Edmonds and Fulkerson~\cite{EdmondsFulkerson1970}.
In the comparison model, Gabow and Tarjan~\cite{GabowTarjan1988} obtained a deterministic $\Oh(m\log^* n)$-time algorithm for directed BP (via their algorithm for directed bottleneck spanning trees).
Chechik, Kaplan, Thorup, Zamir, and Zwick~\cite{ChechikEtAl2016} later gave a slightly improved randomized comparison-based bound and a deterministic linear-time word-RAM algorithm.
They left the existence of a linear-time comparison-based BP algorithm as an open problem.

A \emph{deterministic graphical game} (DGG) is specified by a finite directed graph $G=(V,E)$, a disjoint partition $V=V_0\cup V_1\cup T$, a payoff function $p\colon T\to\Gamma$, and a designated start vertex $s\in V_0\cup V_1$.
Here, $\Gamma$ is a totally ordered payoff set containing a distinguished draw payoff $0$.
Terminals have no outgoing edges, and every nonterminal vertex has at least one outgoing edge.
A token starts at $s$.
At a vertex $u\in V_i$, player $i$ chooses an outgoing edge $(u,v)$ and moves the token to $v$.
If the token reaches a terminal $t$, the outcome is $p(t)$; an infinite play has outcome $0$.
Player~$0$, also known as \Min, tries to minimize the outcome, while player~$1$, also known as \Max, tries to maximize it.
Solving the game from $s$ means finding its value and optimal strategies for both players.
The model was introduced by Washburn~\cite{Washburn1990}.

Andersson, Hansen, Miltersen, and S{\o}rensen~\cite{AnderssonEtAl2010} gave a near-linear-time comparison-based algorithm for solving a DGG from a designated start vertex and asked whether the problem admits a linear-time comparison-based algorithm.
Chechik, Kaplan, Thorup, Zamir, and Zwick~\cite{ChechikEtAl2016} claimed to resolve this open question by giving a simple deterministic linear-time comparison-based algorithm.

We propose a correction to this last result.
First, we give a counterexample to the procedure of~\cite{ChechikEtAl2016}.
Second, we give a deterministic linear-time reduction from directed $s$--$t$ BP to the problem of computing the value of a designated start vertex of a DGG instance.
In particular, the DGG problem is at least as hard as directed $s$--$t$ BP.
These observations leave the following two questions open:

\begin{quote}
\centering{\textit{Does directed $s$--$t$ BP admit a linear-time comparison-based algorithm?}}
\end{quote}

\begin{quote}
\centering{\textit{Does DGG with a designated start vertex admit a linear-time comparison-based algorithm?}}
\end{quote}

Our reduction shows that a positive answer to the second question would imply a positive answer to the first.

\section{The Claimed Algorithm and a Counterexample}
\label{sec:counterexample}

Theorem~9 of Chechik et al.~\cite{ChechikEtAl2016} claims a deterministic $\Oh(m)$-time comparison-based algorithm that, given a DGG and a designated start vertex, finds its value and optimal strategies for both players.
The key step in their algorithm is the following lemma.
We restate the lemma and spell out its proof to make the issue self-contained.

\begin{lemma}[Lemma~8 of Chechik et al.~\cite{ChechikEtAl2016}]
\label{lem:claimed-lemma-eight}
Let $G=(V,E)$ be a DGG with $V=V_0\cup V_1\cup T$, $T=\set{t_1,t_2}$, and $0<p(t_1)<p(t_2)$.
For $i\in\set{1,2}$, let $W_i$ be the set of vertices whose value is $p(t_i)$, and let
\[
  E_i\coloneqq\setof{(u,v)\in E}{v\in W_i},
  \qquad m_i\coloneqq|E_i|.
\]
If $W_1\cup W_2=V$ (that is, no vertex has value $0$), then there is a deterministic algorithm for computing either $W_1$ or $W_2$ in $\Oh(\min\set{m_1,m_2})$ time.
\end{lemma}

\begin{proof}[Proof given by Chechik et al.]
Their Lemma~7 gives a backward-search algorithm for a game with a unique positive-payoff terminal.
The search begins at the terminal and processes incoming edges.
When processing $(u,v)$ with $v$ already discovered, it discovers $u$ if $u$ is a \Max vertex or if $(u,v)$ is the last remaining outgoing edge of the \Min vertex $u$; otherwise it removes $(u,v)$.
Processing one incoming edge is called a basic step.

Run two copies of this search in parallel.
In the first copy, add a self-loop to $t_1$, so that $t_1$ is no longer a terminal; this copy is meant to construct $W_2$.
In the second copy, add a self-loop to $t_2$ and exchange the roles of the two players; this copy is meant to construct $W_1$.
Alternate basic steps between the two copies and stop both as soon as one copy finishes.
Since the two searches are claimed to take $\Oh(m_2)$ and $\Oh(m_1)$ time, respectively, the parallel execution takes $\Oh(\min\{m_1,m_2\})$ time.
\end{proof}

While the presented algorithm for computing $W_2$ is correct, the algorithm for computing $W_1$ does not work in general.
Fix payoffs $0<a<b$ and consider the game in \cref{fig:counterexample}.
Its vertices are $\{s,u,v,t_1,t_2\}$.
The vertices $s$ and $v$ belong to \Min, the vertex $u$ belongs to \Max, and $t_1,t_2$ are terminals with $p(t_1)=a$ and $p(t_2)=b$.
The edges are
\[
  (s,u),\quad (s,v),\quad (u,u),\quad (u,t_1),\quad (v,t_2).
\]
(If no self-loops are allowed, the self-loop of vertex $u$ may be replaced by a cycle.)

\begin{figure}[tb]
  \centering
  \begin{tikzpicture}[
      minvertex/.style={circle,draw,thick,minimum size=8mm,inner sep=0pt},
      maxvertex/.style={rectangle,draw,thick,minimum size=8mm,inner sep=0pt},
      terminal/.style={circle,double,draw,thick,minimum size=8mm,inner sep=0pt},
      gameedge/.style={-{Latex[length=2.2mm]},thick}]
    \node[minvertex] (s) at (0,0) {$s$};
    \node[maxvertex] (u) at (2.25,1.05) {$u$};
    \node[minvertex] (v) at (2.25,-1.05) {$v$};
    \node[terminal] (t1) at (4.75,1.05) {$t_1$};
    \node[terminal] (t2) at (4.75,-1.05) {$t_2$};
    \draw[gameedge] (s) -- (u);
    \draw[gameedge] (s) -- (v);
    \draw[gameedge] (u) -- (t1);
    \draw[gameedge] (v) -- (t2);
    \draw[gameedge] (u) edge[loop above,min distance=9mm] (u);
    \node[right=2mm of t1] {$p(t_1)=a$};
    \node[right=2mm of t2] {$p(t_2)=b$};
    \node[left=1mm of s] {\footnotesize \Min};
    \node[above left=0mm and 1mm of u] {\footnotesize \Max};
    \node[below left=0mm and 1mm of v] {\footnotesize \Min};
  \end{tikzpicture}
  \caption{A counterexample to the parallel procedure in the proof of Lemma~8 of Chechik et al.~\cite{ChechikEtAl2016}.
  Circles represent \Min vertices, the square represents a \Max vertex, and double circles represent terminals.}
  \label{fig:counterexample}
\end{figure}
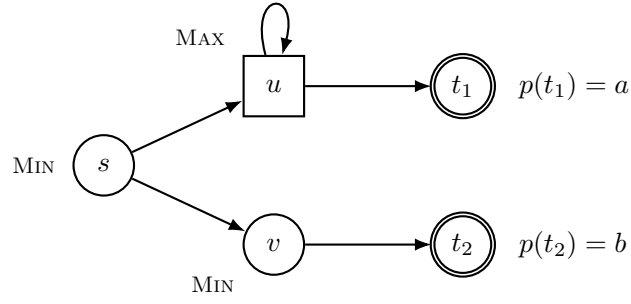

The vertex values are immediate.
We have $\val(v)=b$ because the only move from $v$ leads to $t_2$.
At $u$, \Max chooses between the payoff $a$ at $t_1$ and the payoff $0$ obtained by following the self-loop forever; hence $\val(u)=a$.
Finally, \Min chooses the smaller of the values at $u$ and $v$, so $\val(s)=a$.
Therefore $W_1=\set{s,u,t_1}$ and $W_2=\set{v,t_2}$, and the hypothesis $W_1\cup W_2=V$ of \cref{lem:claimed-lemma-eight} is satisfied.

Consider now the two searches in the claimed proof.
The first search adds a self-loop at $t_1$ and searches backward from $t_2$.
It discovers $v$, since $(v,t_2)$ is the last outgoing edge of the \Min vertex $v$.
It does not discover $s$, since $s$ still has the outgoing edge $(s,u)$.
Thus this search computes the correct set $W_2=\set{v,t_2}$.

The second search adds a self-loop at $t_2$ and exchanges the players.
It starts from $t_1$ and processes the only incoming edge $(u,t_1)$.
After the player exchange, $u$ is a \Min vertex, and $(u,t_1)$ is not its last outgoing edge because the self-loop $(u,u)$ remains.
The search therefore removes $(u,t_1)$ instead of discovering $u$.
Its worklist is now empty, so it terminates and returns $W_1'=\set{t_1}$, rather than the true set $W_1=\set{s,u,t_1}$.
If this search receives the first basic step, the parallel algorithm stops immediately with this incorrect answer.
If the $W_2$ search receives the first step, it only discovers $v$; the following $W_1$ step still terminates with $W_1'$.

\section{A Linear-Time Reduction from BP to the DGG Problem}
\label{sec:reduction}

Let $(G=(V,E),w,s,t)$ be an instance of directed $s$--$t$ BP, where $w\colon E\to\Lambda$ and $\Lambda$ is totally ordered.
Write
\[
  \bp_G(s,t)\coloneqq \min_{P:s\leadsto t}\ \max_{e\in P}w(e).
\]

We extend $\Lambda$ by two fresh values, denoted $\bot$ and $\mathbf 0$, ordered as
\begin{equation}
  \bot < \lambda < \mathbf 0
  \qquad\text{for every }\lambda\in\Lambda.
  \label{eq:extended-order}
\end{equation}
The fresh value $\mathbf 0$ is the draw payoff assigned to an infinite play.

If the original BP weights are real and one insists that all game payoffs be real numbers, let $M\coloneqq \max_{e\in E}w(e)$ and replace every weight by
\[
  \bar w(e)\coloneqq w(e)-M-1<0.
\]
This translation preserves every optimal path and shifts the bottleneck value by $-(M+1)$.
Setting $\bot\coloneqq \min_{e\in E}\bar w(e)-1$ then realizes \eqref{eq:extended-order} over the reals, with the usual draw payoff $0$; the original bottleneck value is recovered by adding $M+1$.

We first delete every edge leaving $t$ and discard every vertex that is not both reachable from $s$ and able to reach $t$.
Two graph searches perform this preprocessing in $\Oh(n+m)$ time.
The preprocessing does not change the bottleneck value, and every remaining vertex other than $t$ has an outgoing edge in the new graph.
For notational simplicity, continue to call the resulting instance $G=(V,E)$.

For every edge $e\in E$, create a vertex $x_e$ and a terminal $z_e$.
The vertices of the game $\mathcal G(G)$ are
\[
  V\ \cup \ X\ \cup \ Z,
  \qquad
  X\coloneqq \setof{x_e}{e\in E},\quad
  Z\coloneqq \setof{z_e}{e\in E}.
\]
The vertex $t$ and all vertices in $Z$ are terminals.
Every vertex in $V\setminus\{t\}$ belongs to \Min, and every vertex in $X$ belongs to \Max.
For each $e=(u,v) \in E$, add the three edges
\[
  (u,x_e),\qquad (x_e,v),\qquad (x_e,z_e)
\]
to the game.
Finally, set $p(t)=\bot$ and $p(z_e)=w(e)$.
An infinite play has payoff $\mathbf 0$, which is larger than every edge weight.
The edge gadget is shown in \cref{fig:reduction-gadget}.
The game has at most $n+2m$ vertices and $3m$ edges.

\begin{figure}[tb]
  \centering
  \begin{tikzpicture}[
      minvertex/.style={circle,draw,thick,minimum size=7mm,inner sep=0pt},
      maxvertex/.style={rectangle,draw,thick,minimum size=7mm,inner sep=0pt},
      terminal/.style={circle,double,draw,thick,minimum size=7mm,inner sep=0pt},
      gameedge/.style={-{Latex[length=2.2mm]},thick},
      node distance=19mm and 24mm]
    \node[minvertex] (u) {$u$};
    \node[maxvertex,right=of u] (xe) {$x_e$};
    \node[minvertex,right=of xe] (v) {$v$};
    \node[terminal,below=13mm of xe] (ze) {$z_e$};
    \draw[gameedge] (u) -- (xe);
    \draw[gameedge] (xe) -- (v);
    \draw[gameedge] (xe) -- (ze);
    \node[below=1mm of ze] {$p(z_e)=w(e)$};
    \node[below=1mm of u] {\footnotesize \Min};
    \node[above=1mm of xe] {\footnotesize \Max};
    \node[below=1mm of v] {\footnotesize \Min};
  \end{tikzpicture}
  \caption{The DGG gadget replacing an edge $e=(u,v)$.}
  \label{fig:reduction-gadget}
\end{figure}
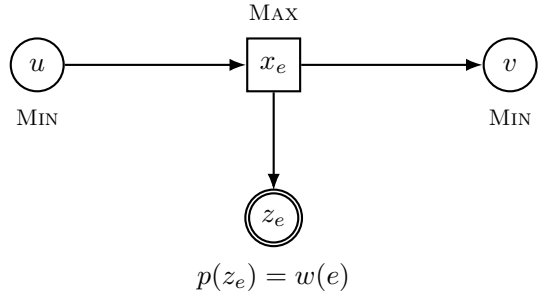

\begin{theorem}
\label{thm:reduction}
The value of $s$ in $\mathcal G(G)$ is $\bp_G(s,t)$.

\end{theorem}

\begin{proof}
Write $B\coloneqq \bp_G(s,t)$.
We prove matching upper and lower bounds on the game value.

For the upper bound, fix a simple $s$--$t$ path $P$ whose bottleneck is $B$.
At each original vertex on $P$, \Min chooses the vertex $x_e$ belonging to the next edge $e$ of $P$.
At $x_e$, \Max can stop at $z_e$, obtaining $w(e)\leq B$, or continue to the next original vertex.
If \Max always continues, the play reaches $t$ and obtains $\bot<B$.
Hence this \Min strategy guarantees an outcome at most $B$, so $\val_{\mathcal G(G)}(s)\leq B$.

For the lower bound, let \Max stop at $z_e$ whenever $w(e)\geq B$, and otherwise continue to the end-vertex of $e$.
A play that stops has payoff at least $B$, and an infinite play has payoff $\mathbf 0>B$.
If a play reached $t$ without stopping, its projection in $G$ would be an $s$--$t$ walk using only edges of weight strictly smaller than $B$.
Removing cycles would give an $s$--$t$ path with bottleneck smaller than $B$, a contradiction.
Thus \Max guarantees an outcome at least $B$, and $\val_{\mathcal G(G)}(s)\geq B$.
\end{proof}

\begin{corollary}
If the value of a designated start vertex in a DGG can be computed (deterministically) in linear time in the comparison model, then directed $s$--$t$ BP can be solved (deterministically) in $\Oh(n+m)$ time in the comparison model.
\end{corollary}

\begin{proof}
Apply \cref{thm:reduction}.
The game $\mathcal G$ can be constructed in linear time.
If an optimal path rather than only its value is required, perform a graph search in the subgraph consisting of the edges whose weight is at most $\bp_G(s,t)$.
\end{proof}

\subparagraph{AI Disclosure.}
Large language models were used to prepare an initial draft from proof notes written by the author and to assist with figure creation.
The author subsequently revised and verified the manuscript and takes full responsibility for its contents.
All mathematical ideas and results are the author's own.

\bibliography{refs}

\end{document}